\documentclass[letterpaper, 11pt]{article}  % Comment this line out if you need a4paper

\usepackage{graphicx}
\usepackage{amssymb}
\usepackage{amsmath}
\usepackage{algorithmic,algorithm}
\usepackage{subfig}
\usepackage{dsfont}
\usepackage[margin=1in]{geometry}
\def\qed{\hfill\ensuremath{\square}}

\def\cd_lattice{(1)} % Added for compatibility - Substituted all tikz-cd with figures

\newtheorem{remark}{Remark}[section]

\newtheorem{lemma}{Lemma}[section]
\newtheorem{proposition}{Proposition}[section]
\newtheorem{definition}{Definition}[section]

\newenvironment{proof}{\noindent\emph{Proof:}}{\qed}
\def\k{$k$}

\def\RR{\mathds{R}}
\def\D{\mathcal{D}}

\title{An Algebraic Topological Perspective to Privacy:\\ Numerical and Categorical Data\thanks{This work was supported by United Technologies Research Center.}}

\author{Alberto Speranzon\hspace*{1cm}\and Shaunak D. Bopardikar\thanks{Alberto Speranzon was with United Technologies Research Center at the time when this work was developed. He is now with Honeywell Aerospace -- Advanced Technology, \texttt{alberto.speranzon@honeywell.com}. Shaunak D. Bopardikar is with United Technologies Research Center, Inc. \texttt{bopardsd@utrc.utc.com}.}}

\begin{document}

\maketitle

\begin{abstract}
In this paper, we cast the classic problem of achieving \k-anonymity for a given database as a problem in algebraic topology. Using techniques from this field of mathematics, we propose a framework for \k-anonymity that brings new insights and algorithms to anonymize a database. We begin by addressing the simpler case when the data lies in a metric space. This case is instrumental to introduce the main ideas and notation. Specifically, by mapping a database to the Euclidean space and by considering the distance between datapoints, we introduce a simplicial representation of the data and show how concepts from algebraic topology, such as the nerve complex and persistent homology, can be applied to efficiently obtain the entire spectrum of \k-anonymity of the database for various values of~\k~and levels of generalization. For this representation, we provide an analytic characterization of conditions under which a given representation of the dataset is \k-anonymous. We introduce a weighted barcode diagram which, in this context, becomes a computational tool to tradeoff data anonymity with data loss expressed as level of generalization. Some simulations results are used to illustrate the main idea of the paper. We conclude the paper with a discussion on how to extend this method to address the general case of a mix of categorical and metric data. 
\end{abstract}

\section{Introduction}

Recent times have seen a revolution in computing technologies. Third-party computing services, such as the Cloud, have been creating new paradigms for both, data storage and computation. Such technologies require one to repeatedly revisit the basic question of ``How do we protect the data from a privacy perspective?". Although this problem originated in database theory and computer science applications, it has recently expanded to domains such as systems and control. In control systems, there is always an information flow between sensors/actuators and controllers/plants and even between controllers in the case of distributed or cloud-based control. In critical cyber-physical systems such as power or transportation networks, where information about users and utility companies are exchanged or information about multiple users is fused, privacy has become a critical element to be considered~\cite{TerraSwarm}. Clearly, as the world is becoming more connected and loops are increasingly being closed over millions of sensors~\cite{Vincentelli:15}, privacy is becoming a top priority in the control community.

In this paper, we consider static data collected within a database that, in the context of cyber-physical systems, could represent log/monitoring data that needs to be analyzed offline to determine overall system performance, enterprise level fault detection and propagation, forensic analysis, etc. Although several approaches have been proposed to address different aspects of the privatization of data within a database, this paper provides a novel framework and perspective to address the most classical version of the problem using concepts and tools from algebraic topology.

\subsection{Literature}

Among several methods that have been developed for database privacy, a classic and popular approach is \emph{\k-anonymity}, which is a mechanism for protecting privacy of individuals represented as entries in a database~\cite{sweeney2002}. The idea consists of removing all attributes from a database that can be used as unique identifiers but retain a set of attributes, called \emph{quasi-identifiers}, for which identification may be possible but these are necessary for analysis. To such quasi-identifiers -- ZIP code, date of birth, GPS location, energy usage, data/time,  etc. -- one applies a transformation that generalizes their value so that data records become indistinguishable. More precisely, for a given value of~\k, the original database is modified so that at least~\k individuals in the database have identical quasi-identifiers. This is achieved by generalizing numeric or text attributes: for example, the ZIP code can be generalized so that a certain number of the least significant digits are suppressed as $46532 \rightarrow 465\!*\!*$, age could be generalized to intervals, $35 \rightarrow [30, 40]$, and the gender could be generalized as $\{M,F\} \rightarrow Person$.

The problem of computing an optimal \k-anonymous version of a database has been shown to be NP-hard~\cite{meyerson2004}. However, efficient algorithms such as Incognito~\cite{lefevre2005} and its variants or greedy clustering-based algorithms~\cite{byun2007} have been proposed to achieve \k-anonymity. A multi-dimensional extension of the greedy approaches has been addressed in~\cite{lefevre2006}, which results into a representation of the database that is reminiscent of classic grid-based paintings by Mondrian. In multi-dimensional settings, a data aggregation scheme based on Hilbert curves has been proposed in~\cite{kim2013}. 

Algebraic topology is a branch of mathematics that leverages tools and concepts from abstract algebra to study topological spaces. For example, a simple model for sensor network is a set of points in a (multidimensional) space in which two points (or sensors) are \emph{neighbours} if they are within a specified distance of each other. Then, concepts from algebraic topology, such as homology, have been used to detect \emph{holes} in sensor networks~\cite{ghrist2005,desilva2007}. Distributed algorithms to localize holes in sensor networks using related concepts have been addressed in~\cite{muhammad2006} and in~\cite{tahbazsalehi2010}. Recently such methods have been also used for filtering and position estimation in~\cite{JD-AS-RG:13,RG-DL-JD-AS:12}.

One concept of privacy which has been applied to several control problems is that of differential privacy~\cite{CD-AR:13}. Informally, this concept means that for a given database, if any single individual is removed from the database, then no output of a computation run on the database would become significantly more or less likely. This concept has been applied to achieve differential privacy of Kalman filtering and estimation problems~\cite{JLN-GJP:14}, to ensure a level of truthfulness in electric vehicle charging applications \cite{SH-UT-GJP:15}, to achieve average consensus in a private manner \cite{YM-RM:14}, to name a few. While the concept of differential privacy is very general and is applicable to dynamic databases, the resulting mechanism relies very strongly on the type of function/query that needs to be computed on the database. In contrast, \k-anonymity is a static concept but is independent of any computation to be carried out on the database and therefore suitable in the context of offline analysis. That said, there are situations when \k-anonymity is not sufficient and individuals can be re-identified despite anonymization. Approaches to deal with such cases have been considered such as $\ell$-diversity~\cite{ldiversity2007} and $t$-closeness~\cite{tcloseness2007}. Even though such advanced concepts have been proposed and privacy metrics continue to be an active research area, \k-anonymity is still widely used. 

Extensions of the proposed approach to other metrics will be a subject of future study.

\subsection{Contributions}
This paper introduces a novel perspective to data privacy based on algebraic topology. In particular, we address the case when the data lies in a metric space. By defining two datapoints that lie within a specified radius as neighbours, we show that the representation falls within the natural setting of a \v{C}ech (or in general, a Nerve) complex~\cite{ghrist2014}. By increasing the radius (generalization), we show that the sequence of \v{C}ech complexes result into a \emph{filtration}, i.e. nested complexes. This further implies that tools such as persistent homology can be applied to efficiently obtain the entire spectrum of \k-anonymity of the database for various values of the generalization. The benefit of this approach is that once the family of complexes is built, for various generalization values, we can apply fast and scalable persistent homology algorithms, such as Perseus~\cite{Perseus:15} to determine the tradeoffs. Furthermore, the persistent diagram not only provides the tradeoffs between a generalization and the value of $k$, but also show how many equivalent classes are formed for a given generalization achieving a certain~\k-anonymity, a metric that has an impact on the anonymized data quality~\cite{RD-IR-DW:08}.

For this representation, we provide an analytic characterization of conditions under which a given representation of the dataset is \k-anonymous. Finally, we discuss how this method can be extended to address the general case of a mix of categorical and metric data. 
  
\subsection{Organization of this paper}
This paper is organized as follows. The problem formulation is presented in Section~\ref{sec:problem}. Background results and concepts from algebraic topology 
are reviewed in Section~\ref{sec:topology}. The proposed approach is presented in Section~\ref{sec:approach1} for numerical data along with some simulation results. Extension of the method described in Section~\ref{sec:approach1} to mixed categorical and numerical data is discussed in Section~\ref{sec:approach2}. 

\section{Problem formulation}\label{sec:problem}

Let us consider a database table $T(A_1,A_2,\dots,A_m)$ consisting of~$N$ rows, where each $A_i \in \D$ are various attributes that in general can take the form of \emph{numeric} and/or \emph{categorical} values, i.e., the domain $\D$ can either be a set of discrete or continuous values. Without loss of generality, we can identify with $Q_T = \{A_1,A_2,\dots,A_d\}$ a set of~$d$ attributes that we define as \emph{quasi-identifiers}, namely attributes that can be joined with external information/databases so that private information can be obtained. Typical examples of private data that could be obtained are names of individuals, salaries, etc.

Another database table $\bar{T}(\bar{A}_1, \bar{A}_2, \dots, \bar{A}_m)$ consisting of~$N$ rows is said to be a \emph{generalization} $\bar{T} = G(T)$ of the table~$T$ if, for every row $T_j$ of $T$, 
$$
	Q_{T_j} \subset Q_{\bar{T}_j}\,.
$$
In this paper, as said previously, we will be focusing on the concept of \k-anonymity for privacy, that is formally defined as follows.

\begin{definition}[\k-anonymity \cite{sweeney2002}]
	Consider a generalized database~$\bar{T}$ and a quasi-identifier set~$Q_{\bar{T}}$. The set~$Q_{\bar{T}}$ is said to have the \k-anonymity property if and only if each unique tuple in the projection of~$\bar{T}$ on~$Q_{\bar{T}}$ occurs at least~$k$ times in $\bar{T}$.
\end{definition}

Given a database~$T$, the problem of \k-anonymity is thus to determine a generalization function~$G(.)$ so that the resulting database $\bar{T} = G(T)$ is \k-anonymous. Clearly, one may simply generalize every entry and find the smallest set that generalizes every row of~$T$. However, this trivial method would completely destroy the information content in the original database. The problem is how to minimize such as \emph{over-generalization} of the quasi-identifiers. This notion will be made precise in Section~\ref{sec:approach1}. 

\section{Background on Topological Methods} \label{sec:topology}

In this section, we provide a summary of some useful concepts from algebraic topology that will be used in our approach. Interested readers may refer to~\cite{ghrist2014} and~\cite{ghrist2008} for additional details on these concepts.

\begin{definition}[\v{C}ech complex]
Given a collection of points $\{x_i\}\in\mathbb{R}^n$, the \emph{\v{C}ech complex} is the abstract simplicial complex whose \k-simplices are determined by unordered $(k+1)-$tuples of points $\{x_i\}_0^k$ whose closed $\epsilon$-ball neighborhoods have a point of common intersection.
\end{definition}

\subsection{Simplicial Homology}
	
Homology is an algebraic characterization of ``holes'' in a topological space. The central notion is that of a boundary homomorphism, which in the context of simplicial complexes, encodes how simplices are attached to their lower dimensional facets. To define (simplicial) homology, of a complex~$C$, we choose an ordering of each simplex, in the same way directed graphs are ordered. Given such ordering we consider~$\mathbb{R}$-vector spaces $\mathcal{C}_k(C)$ with basis the oriented $k$-simplices. We thus have that $\mathcal{C}_\bullet$, forms a sequence of vector spaces, which we call chain complex. A \emph{boundary homomorphism} is defined as the linear map $\partial_k:\mathcal{C}_k(C) \rightarrow \mathcal{C}_{k-1}(C)$ given by associating each basis element of~$\mathcal{C}_k(C)$ to the formal sum of its (oriented) faces of dimension~$k-1$. The boundary operator~$\partial=\{\partial_k\}$ thus encodes the assembly instructions of~$C$. It turns out that the~$k^{th}$ homology group of the complex~$C$, $H_k(C)$ is given by
$$
	H_k(C) = Z_k / B_k = \mathrm{ker}\: \partial_k/ \mathrm{im}\: \partial_{k+1}\,.
$$
The group~$Z_k =  \mathrm{ker}\: \partial_k$ is called the~$k$-th cycle group and its elements (chains) represent $k$-cycles. We have that $B_k = \mathrm{im}\: \partial_{k+1}$ is the $k$-th boundary group whose elements are $k$-boundaries. The quotient space~$H_k(C)$ thus represents all the $k$-cycles that are not boundaries of~$k+1$ simplices, namely cycles that represent~$k$-dimensional ``holes''. The homology of the complex~$C$ is then $H_\bullet(C) = \{H_k(C)\}$. 
	
In this paper, we will use the notion of dimension of the $k$-th homology, that is the dimension of the vector space~$H_k(C)$, $\mathrm{dim} H_k(C)$. In particular, the $k$-th homology group $H_k(C)$ is said to be \emph{trivial} if $\mathrm{dim} H_k(C) = 0$.
	
\subsection{Persistent Homology}

Let us consider a sequence of complexes $C^\epsilon$ with $\epsilon = \{\epsilon_1,\epsilon_2, \dots, \epsilon_M\}$, more specifically the sequence of \v{C}ech complexes $\{C^{\epsilon_i}\}_{i=1}^N$, for increasing $\epsilon_i \in \mathbb{R}_{\geq 0}$, $\epsilon_i \leq \epsilon_j$ for $i\neq j$. There are clearly inclusion maps between such complexes
$$
	C^{\epsilon_1} \stackrel{\imath}{\hookrightarrow} C^{\epsilon_2} \stackrel{\imath}{\hookrightarrow} \cdots \stackrel{\imath}{\hookrightarrow} C^{\epsilon_{M-1}} \stackrel{\imath}{\hookrightarrow} C^{\epsilon_M}\,.
$$
Rather than studying the homology of each complex for each value of the parameter~$\epsilon$, one can then study the homology of the inclusions $\imath: H_\bullet(C^{\epsilon_i}) \rightarrow H_\bullet(C^{\epsilon_j})$ for $i < j$. Such maps are important as they capture topological features that persist over the parameter space. 
	
\begin{figure}[!t]
	\centering
	\includegraphics[width=\hsize]{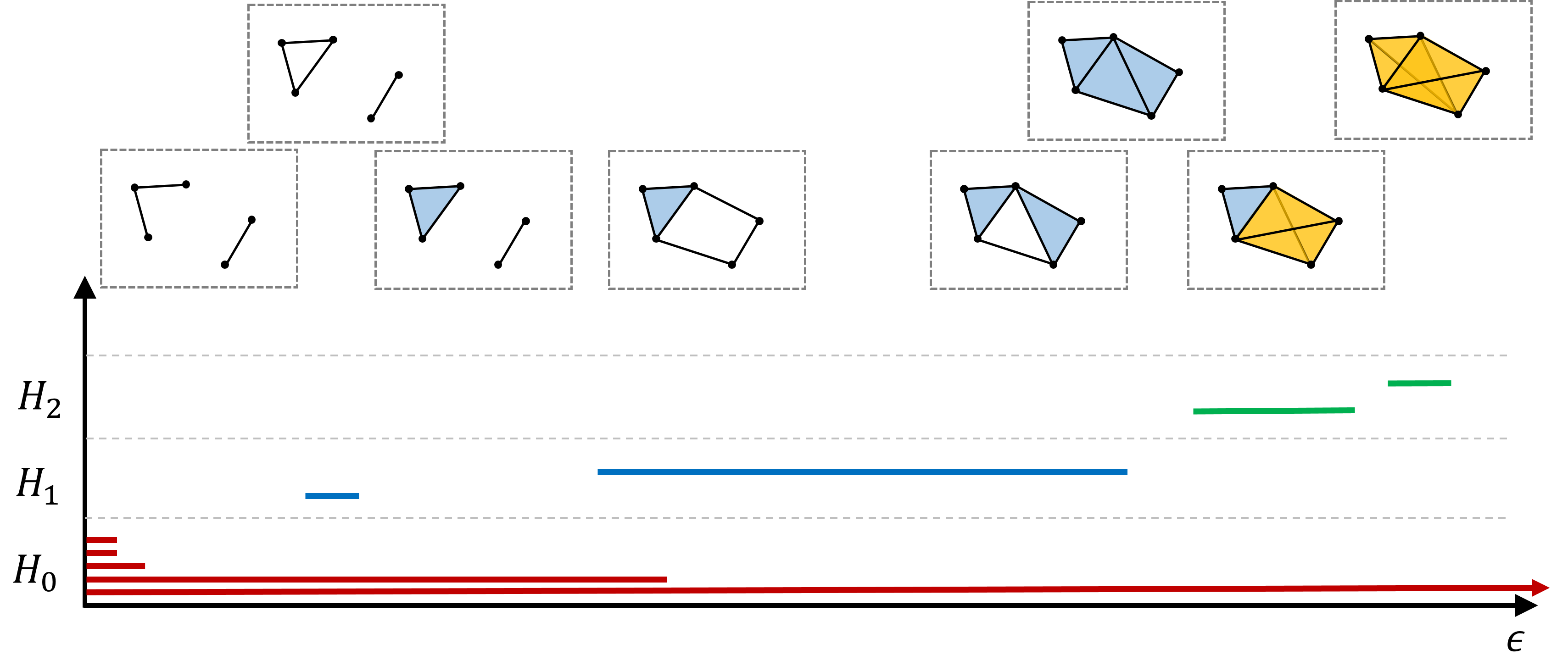}
	\caption{Example of the barcodes.}\label{fig:ex_persistent}
\end{figure}	
	
The dimension of the homology groups as a function of the single parameter~$\epsilon$ can be plotted in a diagram, called the \emph{barcodes diagram}, see~\cite{ghrist2008}. We show a simple example in Figure~\ref{fig:ex_persistent} where $\epsilon$ is the radius of a ball around each vertex, and where there is a $k$-simplex whenever~$k+1$ circles have non-empty intersections. For small values of~$\epsilon$, we have that the number of connected components, $\dim H_0$, is the number of vertices (0-simplices), and as~$\epsilon$ increases, more vertices get connected, resulting in components to merge till a single connected component is obtained. For a small value of~$\epsilon$ there are no high dimensional holes, however, at some point, before the first 2-simplex (blue triangle) gets filled in, such 2-simplex is not filled and generates a hole that quickly disappears. At a later value of~$\epsilon$, a large hole is formed at the time when the vertices form a single large connected component. As the~$\epsilon$ parameter increases further at some point the ``middle'' hole gets filled and the dimension of the first homology, $\dim H_1$ becomes zero again. As~$\epsilon$ further increases, tetrahedrons appear, first with an empty volume, namely a void, that disappear. Higher dimensional holes will likely occur, but we did not depict them. As $\epsilon$ further increases the complex will have trivial higher homology groups and only have a single connected component.
	
In this paper we will make use the persistent homology in the context of \k-anonymity. 
	
\section{$k$-Anonymity via Persistent Homology:\\ Numerical Attributes} \label{sec:approach1}

In this section, we describe the algebraic topological approach toward achieving \k-anonymity. In this section, we will restrict our attention to the case of the attributes in the quasi-identifier set~$Q_T$ being all numeric/continuous variables, such as Age, Salaries, Taxes paid in a year, etc. In this case, we can clearly represent the set~$Q_T$ as a~$|Q_T|$-dimensional vector. We will assume that all the vectors are elements of a real vector space.

\begin{figure}[t!]
	\centering
	\includegraphics[width=0.7\hsize]{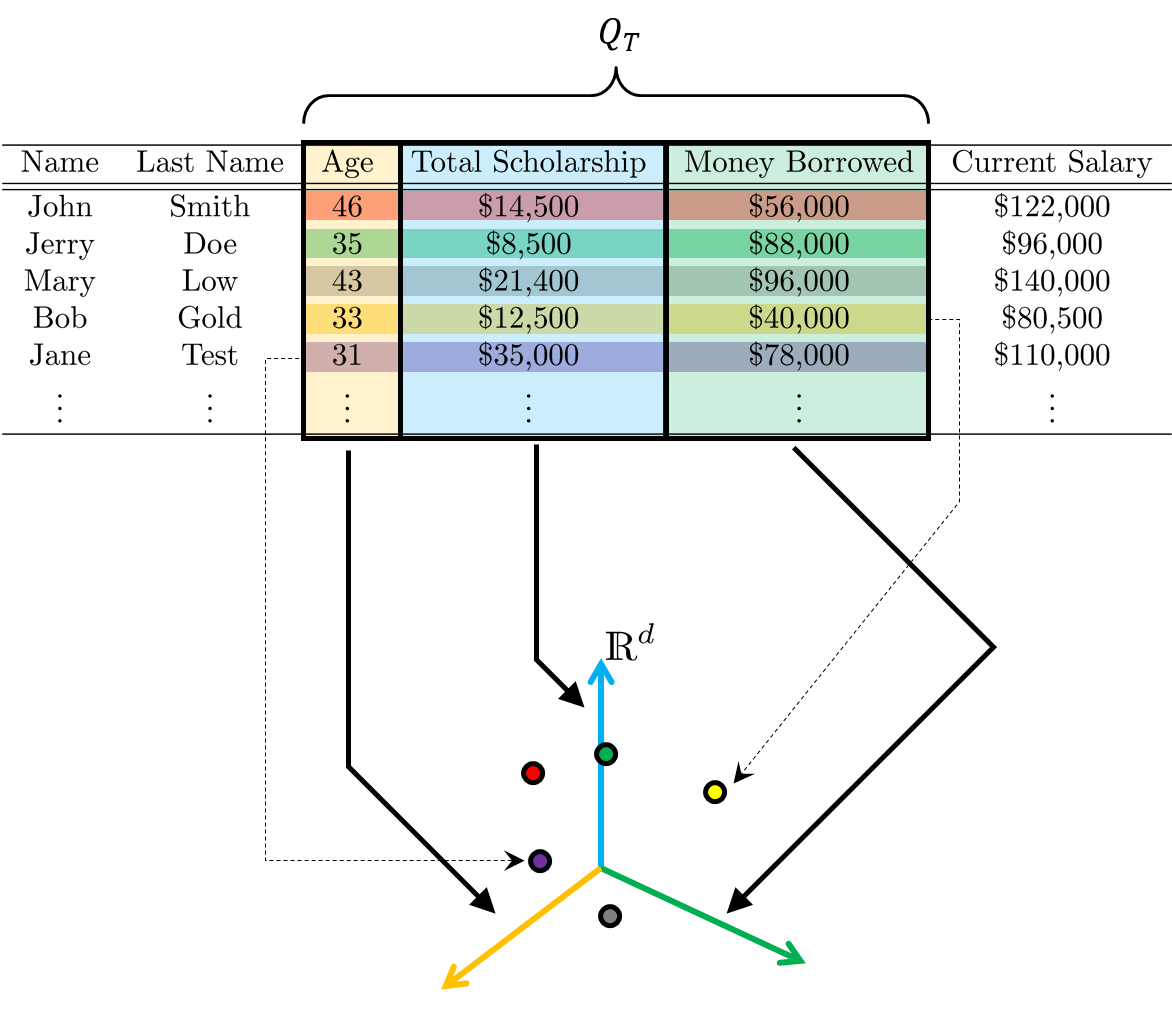}
	\caption{Example of a table where the first two columns are identifiers, the last column is the sensitive data and the middle three are the quasi-identifiers. Each entry in the table can be mapped, through the quasi-identifiers, to a point in the three-dimensional vector space,~$\RR^3$.}\label{fig:table_example}
\end{figure}

Figure~\ref{fig:table_example} shows an example of a table~$T$ with two identifiers (Name and Last Name), three quasi-identifiers $Q_T =\allowbreak \{\texttt{Age},\allowbreak \texttt{Total\space Scholarship},\allowbreak \texttt{Money\space Borrowed}\}$ and the sensitive column  \texttt{Current Salary}. As we show at the bottom of Figure~\ref{fig:table_example}, we can map the quasi-identifiers to a three-dimensional vector space (three dimensional as $|Q_T| = 3$), where each entry in the table corresponds to a point in~$\RR^3$.

In the following, we will often refer to entries of the table~$T$ as \emph{points} due to the previously described representation. Also, note that without any loss of generality, we can consider the data to take values within the hypercube~$\mathcal{M} = [0,1]^{|Q_T|}$, since one can always normalize the data accordingly.

The first definition is a direct application of the \v{C}ech complex as summarized in Section~\ref{sec:topology}. In this paper, we use this structure to capture the \k-anonymity property of the data. 

\begin{definition}[Anonymity Complex]
	Given a table~$T$ with~$N$ rows and a set of quasi-identifier~$Q_T$, let us consider the~$N$ points $\{p_i\}_1^N \in \mathcal{M}^{N}$. We define an \emph{anonymity complex}~$\mathcal{C}(p)$ the simplicial complex whose \k-simplices are determined by $(k+1)$ points $\{p_{i_0}, p_{i_1}, \dots, p_{i_k}\}_0^k$ such that closed $\epsilon$-ball neighborhoods centered around these points have at least one intersection point. We call the radius~$\epsilon$ the \emph{global generalization} strategy.
\end{definition}

\medskip

We now introduce an important building block.

\begin{definition}[Anonymity \k-simplex]
	Given a global generalization~$\epsilon$, we say that~$k$ points have the \k-anonymity property if all the closed $\epsilon$-ball neighborhoods of the~$k$ points all intersect in at least a point. In this case, we have that the~$k$ points form a \k-simplex, which we term as an \emph{anonymity \k-simplex} and denote with~$S_k$.
\end{definition}

\medskip

Figure~\ref{fig:anonymity_k_simplex_example}(a) shows an example of an anonymity 4-simplex for a given global generalization~$\epsilon$ while Figure~\ref{fig:anonymity_k_simplex_example}(b) shows an example where for the same value of~$\epsilon$, 4-anonymity cannot be achieved. 

\begin{figure}
	\centering
	\includegraphics[width=0.7\hsize,trim=85 125 85 85,clip]{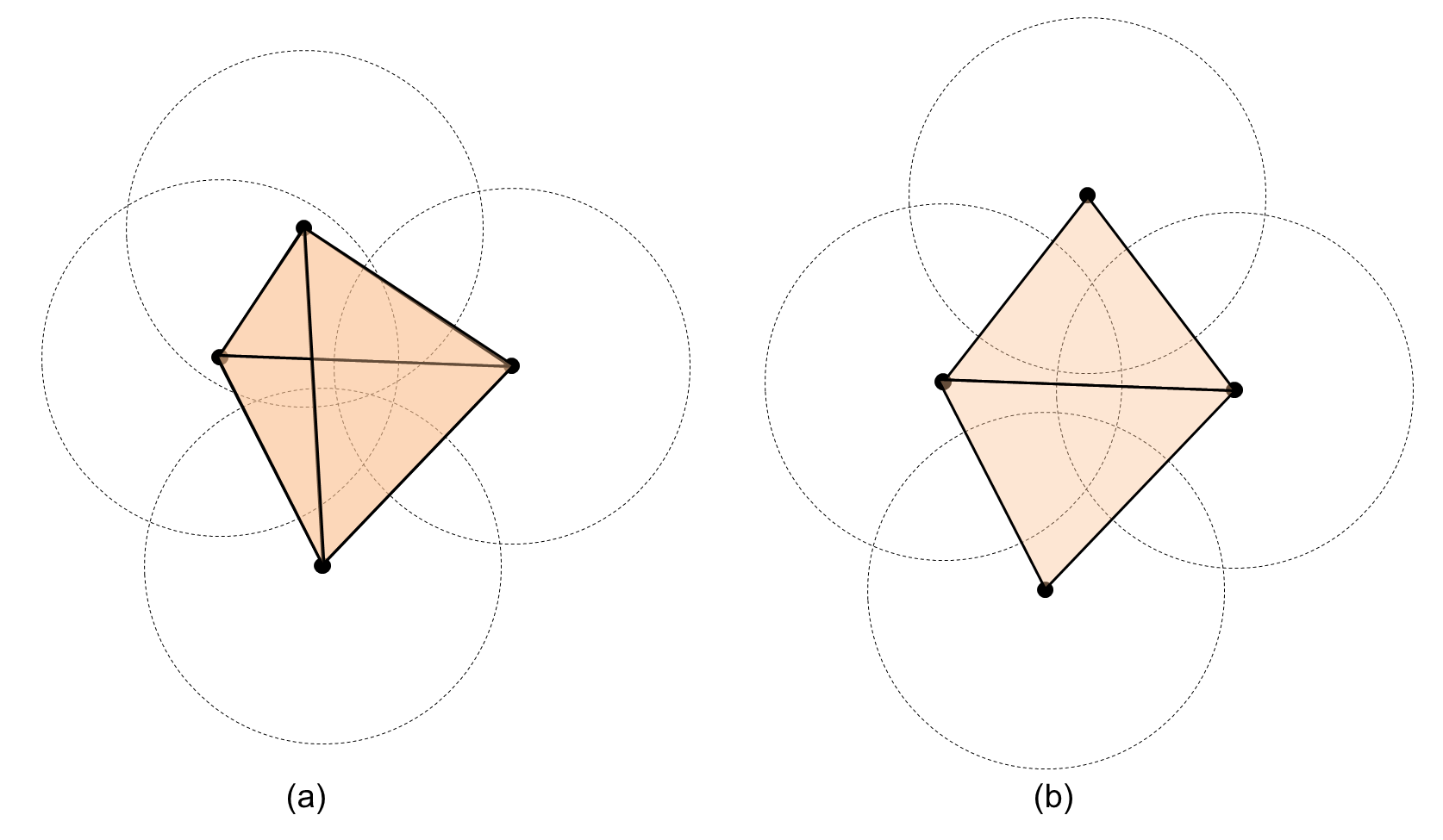}
	\caption{In (a), we show a anonymity 4-simplex representing the fact that there exists an~$\epsilon$ such that the 4 points can be anonymized. In (b), we show an example where there is no anonymity 4-simplex (indeed we have only two 2-simplices) as the selected~$\epsilon$ is not large enough, or equivalently the generalization is high enough, to ensure 4-anonymity.}\label{fig:anonymity_k_simplex_example}
\end{figure}

The following is a useful test to determine whether \k-anonymity can be achieved for a given value of $\epsilon$ or not.

\begin{lemma}\label{(lem:k_anonymity_complex}
	Given a set of points~$p=\{p_i\}_1^N$ corresponding to~$N$ rows of the table~$T$, given a global generalization~$\epsilon$ and the corresponding anonymity complex~$C(p)$, we have that the points~$\{p_i\}_1^N$ have the \k-anonymity property if and only if
	$$
		C(p) = \bigcup_i S_{\ell_i}\,
	$$
	where $S_{\ell_i}$ is the $i$-th anonymity $\ell_i$-simplex with $\ell_i \geq k$ for $i \in \mathds{N}_0$. We say in this case that the anonymity complex achieves \k-anonymity.\medskip
\end{lemma}

\begin{proof}
	The points~$\{p_i\}_1^N$ have the \k-anonymity property if and only if we can sub-divide the set of points into subsets such that
	$$
		\{p_i\}_1^N = \underbrace{\{p_i\}_{1}^{i_1}}_{\Sigma_{\ell_1}} 
		\cup \underbrace{\{p_i\}_{i_1+1}^{i_2}}_{\Sigma_{\ell_2}}\cup \cdots 
		\cup \underbrace{\{p_i\}_{i_{N-1}+1}^N}_{\Sigma_{\ell_N}}\,,
	$$
	and to each subset $\Sigma_{\ell_i}(p)$, we can associate an anonymity $\ell_i$-simplex $S_{\ell_i}$ with $\ell_i \geq k$ for any~$i$ (for given fixed~$\epsilon$) and  $S_{\ell_i} \cap S_{\ell_j} = \emptyset$ for $i\neq j$.
	
	Given the previous set relations, we have that complex~$C$ associated with the set of points~$\{p_i\}$ is then given by the union of the $S_{\ell_i}$. 
\end{proof}

This result then establishes a natural connection between the properties of the anonymity complex, in terms of some of its subcomplexes, and the \k-anonymity property.

We further  explore how topological properties of the anonymity complex are related to the \k-anonymity property and how we can leverage that to find an ``optimal'' generalization. Let us first establish some topological properties of~$C$ and then define explicitly what we mean with ``optimal'' generalization.

%{\bf From SDB: We need the definition of Homology group of C, what it means that the group is trivial, and what is connected component. %Number of equivalence classes, and dimension of homology.}

\begin{proposition}\label{prop:homology_of_C}
	An anonymity complex~$C$, for a given~$\epsilon$, has the \k-anonymity property if and only if its homology groups $H_n(C)$ are trivial for any~$n > 0$, and every connected components is an $\ell$-simplex with $\ell\geq k$. Furthermore, when this is the case the number of equivalence classes generated by using the $\epsilon$ generalization is given by the dimension of the zero-th homology, $\dim H_0(C)$.
\end{proposition}

\begin{proof}
	(If) From Lemma~\ref{(lem:k_anonymity_complex}, we know that we can decompose~$C$ into a finite number of disjoint anonymity \k-simplices. It is known~\cite{ghrist2014} that in this case $H_n(C) = \bigoplus_i H_n(S_{\ell_i})$, namely the $n$-th homology of~$C$ is given by the direct sum of the $n$-th homology of the anonymity \k-simplices. As the \k-simplices are simply connected spaces and contractible, they have trivial high order ($n > 0$) homology groups and $H_0(S_{\ell_i}) \approx \mathds{Z}$ for every~$i$.
	
	(Only if) Let us assume that~$C$ has $H_n(C) = \{0\}$ for $n > 0$, and $H_0(C)$ is non-trivial, and in particular let us assume that $H_0(C) \approx \mathbb{R}^s$. This means that the complex~$C$ has~$s$ connected components. From the hypothesis that the connected components are $\ell$-simplices with $\ell \geq k$, we know that each component is a anonymity simplex. Thus the anonymity complex~$C$ has the \k-anonymity property. 
\end{proof}

We are now able to connect topological properties of the anonymity complex with the \k-anonymity property. Of course, as it can be seen from Proposition~\ref{prop:homology_of_C}, the result still depends on~$\epsilon$, namely the generalization.

When anonymizing a dataset, one is typically interested in ``corrupting'' the data by the least amount. Indeed, if one carefully thinks about the \k-anonymity problem, it is always possible to find a large enough~$k$ that makes the data anonymous, i.e., if one makes the extreme choice of~$k = n$, then the entire data will be in one equivalence class and thus, \k-anonymity will be achieved. The issue with this is that the information contained in the data will be completely lost. 

In the context of this paper, as the generalization is parametrized by~$\epsilon$,  we are interested for find the smallest value of~$\epsilon$ that gives \k-anonymity. We then have the following definition.

\begin{definition}[Minimal Anonymity Complex] \label{def:mincomplex}
	Given an anonymity complex~$C(p)$ associated to a set of points, lets us denote with $C^\epsilon$ the anonymity complex for a given generalization~$\epsilon$. 
	
	We define as \emph{minimal anonymity complex} the following object
	$$
	C^{\epsilon^*} = \min_\epsilon C^\epsilon\,,
	$$
	such that $C^{\epsilon^*}$ achieves \k-anonymity.
\end{definition}

\medskip 

Even without minimization over $\epsilon$, the \k-anonymity problem known to be an NP-hard problem, and so it is clear that we cannot easily find the minimal anonymity complex. To find an approximate solution to the problem, we instead study the \emph{persistent homology of~$C^\epsilon$}. In particular, in this paper, we adapt the idea of persistent homology as a tool to provide the full spectrum of \k-anonymization one can obtain. Our approach is summarized in Algorithm~\ref{algo:ph}.

\begin{algorithm}[!t]
\caption{\k-anonymity via Persistence Homology}
\label{algo:ph}
\begin{algorithmic}
% Initialization
\STATE {\bf Inputs:} $p=\{p_i\}_1^N \in \mathbb{R}^{d\times N}$, Parameter: $k$, Radii: $\{\epsilon_1,\dots, \epsilon_M\}$\\
\STATE Re-scale the dataset into the unit cube $[0,1]^d \subset \mathbb{R}^{d\times N}$.
\FOR{every value of $\epsilon \in \{\epsilon_1,\dots, \epsilon_M\}$}
\STATE Construct the anonymity complex $C^\epsilon(p)$
\ENDFOR
\STATE Compute \emph{weighted} persistent homology
\RETURN Complete bar code diagram
\end{algorithmic}
\end{algorithm}

%The key advantages of the method are:
%\begin{itemize}
%	\item It allow to compute various regimes where data can be \k-anonymized for different values of~$k$ without the need of rerunning the algorithm for every possible value of~$k$;
%	\item It provides directly the tradeoff between~\k and the number of classes;
%	\item Takes advantage of new scalable algorithms {\bf need to define the algorithms} that can deal with very large datasets~\cite{mischaikow2013}.
%\end{itemize}

Note that the parameter~$\epsilon$ induces a family of complexes such that $C^{\epsilon_1} \stackrel{\imath}{\hookrightarrow} C^{\epsilon_2}  \stackrel{\imath}{\hookrightarrow} \cdots \stackrel{\imath}{\hookrightarrow} C^{\epsilon_M}$, where $\epsilon_i \leq \epsilon_j$, for any $i < j$, and thus we recover the same setting as in the persistent homology. Formally we have a $\epsilon$-based filtration~\cite{ghrist2008}. The idea here is to leverage the barcodes or persistent diagram to extract regimes of interests, namely anonymity strategies -- values of $\epsilon$ -- that lead to $k$-anonymity for different values of~$k$.

\begin{table}[!t]
	\centering
	\subfloat[Example data set with $Q_T = \{\mathrm{Age},\mathrm{ZIP\:Code}\}$.]{\begin{tabular}{ccc}
			\hline
			Age & ZIP Code & Salary\\
			\hline\hline
%			John & Doe & 25 & 47677 & \$47,000\\
%			Mary & Low & 22 & 47602 & \$32,000\\
%			Roger & Test & 24 & 47678 & \$52,000\\
%			Lisa & Gold & 43 & 47905 & \$151,000\\
%			Robert & Most & 52 & 47909 &\$145,000\\
%			Jane & Long & 38 & 47906 & \$98,000\\
%			Andrew & Next & 47 & 47605 & \$110,000\\
%			Julie & West & 36 &  47673 & \$92,000\\
%			Greg & Halt & 32 & 47607 & \$115,000\\

			25 & 47677 & \$47,000\\
			22 & 47602 & \$32,000\\
			24 & 47678 & \$52,000\\
			43 & 47905 & \$151,000\\
			52 & 47909 &\$145,000\\
			38 & 47906 & \$98,000\\
			47 & 47605 & \$110,000\\
			36 &  47673 & \$92,000\\
			32 & 47607 & \$115,000\\
			\hline
		\end{tabular}}\hspace{0.4cm} %\vspace{1cm} 
		\subfloat[3-anonymized table]{\begin{tabular}{cc}
				\hline
				Age & ZIP Code\\
				\hline\hline
				{[}22-25{]} & {[}47602-47678{]}\\
				{[}22-25{]} & {[}47602-47678{]}\\
				{[}22-25{]} & {[}47602-47678{]}\\
				{[}38-52{]} & {[}47905-47909{]}\\
				{[}38-52{]} & {[}47905-47909{]}\\
				{[}38-52{]} & {[}47905-47909{]}\\
				{[}32-47{]} & {[}47605-47603{]}\\
				{[}32-47{]} & {[}47605-47603{]}\\
				{[}32-47{]} & {[}47605-47603{]}\\
				\hline
			\end{tabular}}
			\caption{Sample data set for illustrative purpose.}\label{tbl:example_1}
	\end{table}
	
Such a persistent diagram has two specific features:
\begin{itemize}
	\item Each bar in $H_0$ diagram has a \emph{weight} corresponding to the number of elements in the connected components;
	\item Given a value~$k$, we only consider bars that have at least~$k$ elements.
\end{itemize}
	
To illustrate the proposed approach, consider the sample dataset shown in part (a) of Table~\ref{tbl:example_1}. The quasi-identifier set $Q_T = \{\mathrm{Age},\mathrm{ZIP\:Code}\}$, while the salary field is the sensitive one. Figure~\ref{fig:barcode} shows an example of the output of the barcode diagram. For the case of 3-anonymity in particular, the lower center figure depicts $H_0$ while the one in top center shows the $H_1$. We can see that a \emph{hole} gets created for the values of the radius approximately in $[0.11, 0.13]$ and then it gets \emph{filled} thereby making the dataset 3-anonymous. But again in the interval $[0.17, 0.19]$, a hole gets created around when the complex $[1,2,3,7,8,9]$ gets formed. After~$\epsilon$ increases more, this hole disappears and we are left with an anonymity 3-simplex $[4,5,6]$ and an anonymity 6-simplex $[1,2,3,7,8,9]$. 

We can clearly see three regimes where 3-anonymity is possible. The first one (as indicated) has three classes with three elements in each. The second regime corresponds to the values of radius in the interval $[0.19, 0.4]$ which has only two equivalence classes, and the third one (radius greater than $0.4$) being the trivial solution where there is only one class with nine elements.

For minimal data quality loss, the interval $[0.17,0.19]$ is the best solution as 3-anonymity can be reached with largest number of classes. Mapping the first one back to the dataset yields a 3-anonymous version of the dataset shown in part (b) of Table~\ref{tbl:example_1}.

\begin{figure*}[ht!]
	\centering
	\includegraphics[width=\hsize,trim={0.22cm 0.25cm 0.22cm 0.12cm},clip]{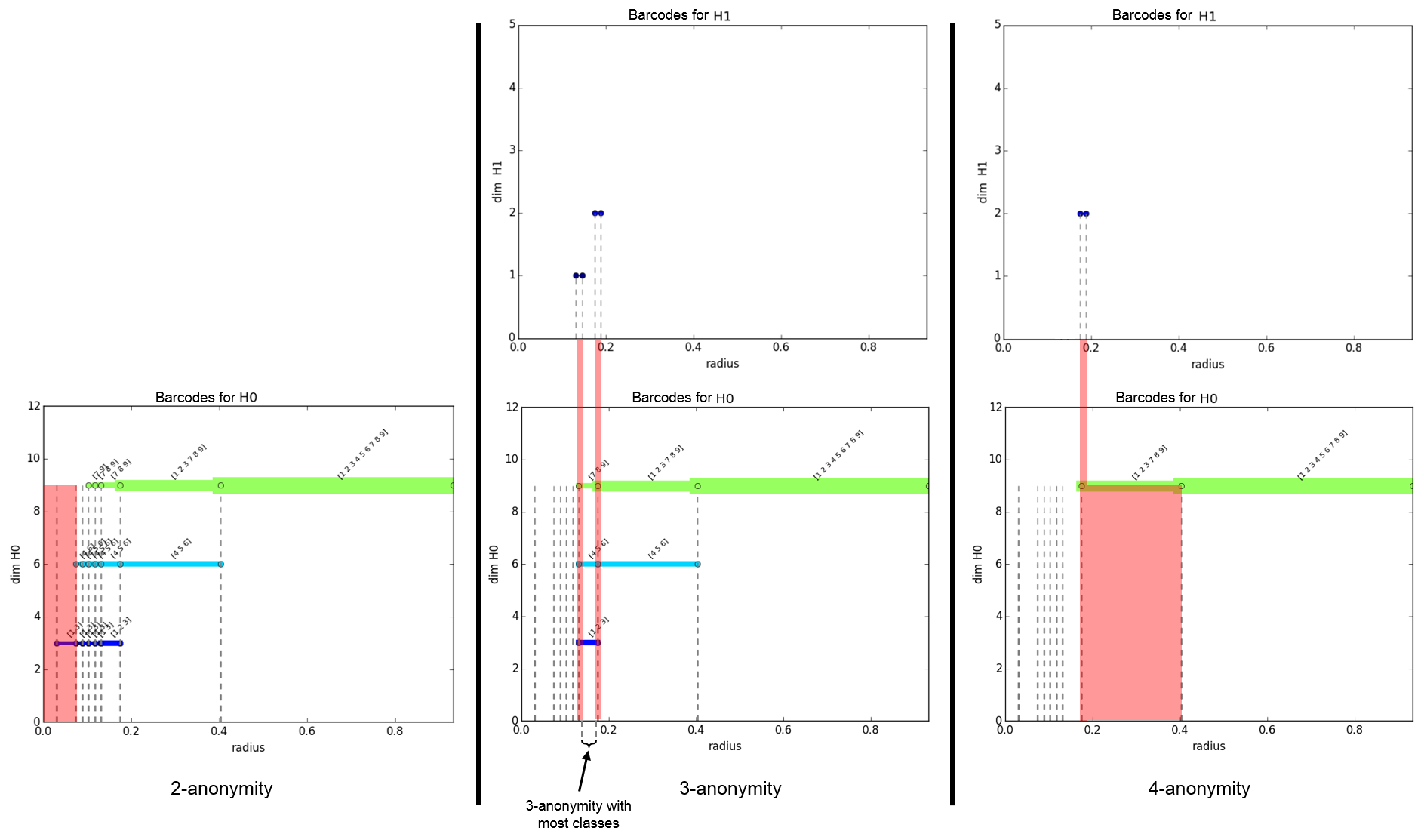}
	\caption{ \small Weighted barcode showing the full spectrum of anonymity regimes. Although the barcode is only a single diagram, we split it here into three diagrams where we only show the simplices that meet the requirements of the \k-anonymity indicated underneath each figure. With red bars we show regimes that cannot achieve 2,3,4-anonymity. These correspond to either situations where not all the $k$-simplices admit trivial higher order homology groups, as is the case for 3-anonymity between $[0.11,0.13]$ and $[0.17,0.19]$, not all the simplices are $\ell$-simplices with $\ell \geq k$. This happens for example in 2-anonymity in the interval $[0,0.08]$ and between $[0.19,0.4]$ for 4-anonymity. We only show here $\dim H_0$ and $\dim H_1$ as higher order homology groups are trivial for this simple example.}\label{fig:barcode}
\end{figure*}

In Figure~\ref{fig:barcode} shows that 2-anonymity can be achieved for generalizations $\epsilon > 0.08$. For generalizations with $\epsilon \leq 0.08$, we have only one anonymity 2-simplex and the rest of the data will be just points, thus 2-anonymity cannot be achieved. The rightmost plot shows that 4-anonymity can be reached for $\epsilon > 0.4$. Note that even if there is an anonymity 6-simplex in the interval $[0.19,0.4]$, there is no way for the $[4,5,6]$ complex to achieve 4-anonymity and thus the generalizations in the interval $[0.19,0.4]$ will not yield 4-anonymity. For $\epsilon > 0.4]$, we can clearly achieve 4-anonymity, but as everything gets into a single class, all the data will be generalized to the same record and thus data quality is compromised.	
		
\begin{remark}[Advantages of proposed approach]\label{rem:adv}
The main advantage of the method proposed is that it enables us not only to find a \k-anonymization for a fixed~$k$ (if it exists), but also to provide alternative regimes that can be especially useful when \k-anonymity cannot be achieved for the given~$k$. This generally is not something that other algorithms, such as~\cite{lefevre2005,lefevre2006,kim2013} directly provide. Indeed, one would need to run the same algorithms for various values of~$k$ to obtain the same tradeoff picture as we obtain. The persistent diagram we consider in this paper is instead computed in \emph{one shot} from the filtration $\{C^\epsilon\}$. Furthermore, we leverage very scalable algorithms for such computation based on discrete Morse theory, see~\cite{mischaikow2013,Perseus:15}. Also, note that not only the persistent diagram allows us to determine the right regime that gives the desired \k-anonymity, but we can also look at the other important tradeoff parameter such as the number of classes. For a given~$k$, it is indeed possible to find various generalizations~$\epsilon$ that meet the \k-anonymity requirement, however some might lead to equivalence classes with many more than~$k$ elements that is in general not  desirable. %\footnote{To better understand this, one can think at the extreme case where all the data is anonymized to such a level that only a single equivalence class exists. In this case we would meet any $k$-anonymity requirement (as long as $k < n$, but the database would be useless as all the data gets reduced to the same value.}. 
\end{remark}

\section{Zig-Zag Persistent Homology for Mixed Data}\label{sec:approach2}

The methodology we have described in the previous section has clearly nice properties but has some limitations. First, it is restricted to numerical attributes where the notion of a radius is well defined and thus it would not work in the case where attributes are categorical in nature, such as, for example, strings, labels, social security numbers, etc. Second, growing balls (or polytopic approximations) in high dimensional spaces and determining intersections can be computationally challenging. In this section, we discuss an extension that leads to weighted persistent diagrams as the one described in the previous section and that can be used to explore anonymity tradeoffs. 
	
\begin{figure*}[t!]
	\centering
	\subfloat[Example of generalization for two set of labels. On the top for \emph{countries} and on the bottom for \emph{gender}.]{\includegraphics[width=0.65\hsize]{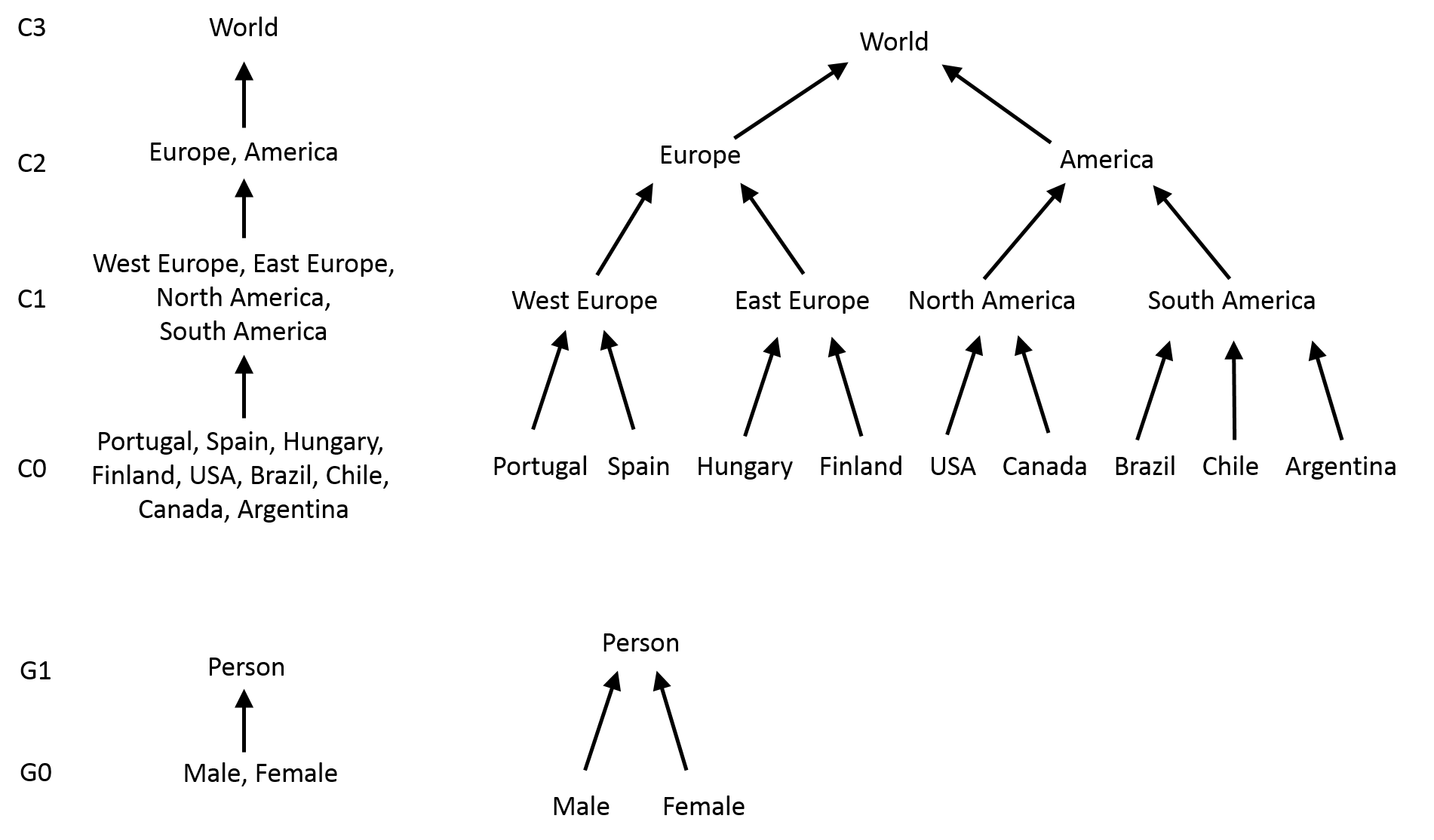}\label{fig:ex_generalizations}}\hfill
	\subfloat[Generalization lattice for the attributes \emph{countries} and  \emph{gender}.]{\includegraphics[width=0.31\hsize]{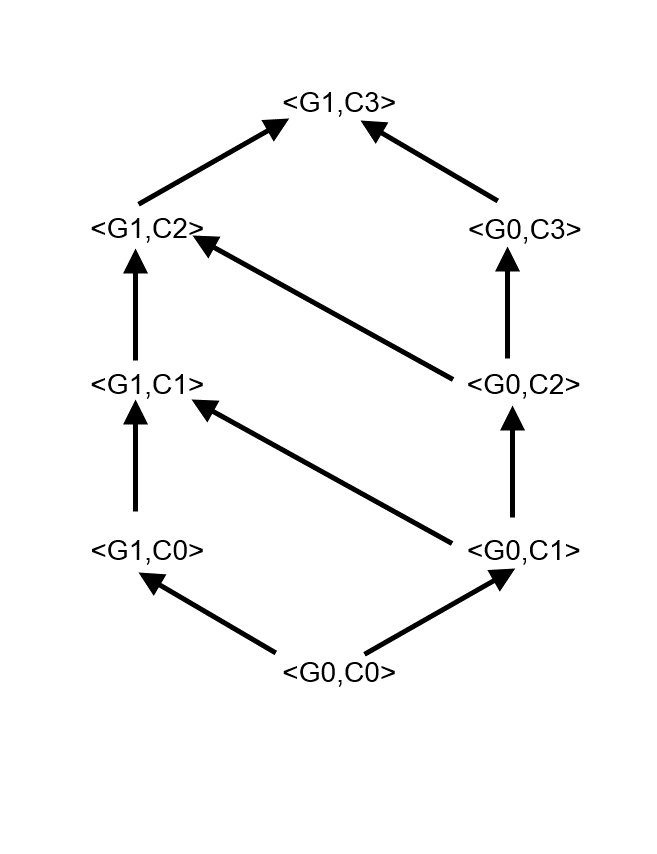}\label{fig:ex_lattice}}
	\caption{In (a) we show  two generalization trees and in (b) we show the corresponding generalization lattice.}
\end{figure*}	
	
For anonymization of categorical data, one needs to specify \emph{generalization trees} that enforce a certain partial order between the various generalizations. Let us consider a simple example where there are two attributes: $Countries$ and $Gender$. Examples of generalization tree are shown in Figure~\ref{fig:ex_generalizations} where, as one moves from leaves to the root, the generalization becomes increasingly coarser, see for example~\cite{lefevre2005}. We thus have that $\{Portugal,Spain\} \prec \{West\,Europe\}$ and $\{West\,Europe, East\,Europe\} \prec \{Europe\}$, etc. or where $\{Male, Female\} \prec \{Person\}$. We may capture the generalization of numerical values within trees as well.

Formally, we have that a generalization tree is a map $\mathcal{T}:\mathcal{N}\times[0,r] \rightarrow \mathcal{A}$ where $\mathcal{N}$ is the set all nodes in the trees (including root and leaves) and $[0,r]$ is the level of the generalization. For example, $\mathcal{T}(\{USA,Canada\},2) = America$. For this setting, we can extend the notion of anonymity complex as follows.
	
\begin{definition}[Generalized Anonymity Complex]
	Given a table~$T$ with~$N$ rows and a set of quasi-identifiers~$Q_T$ as well as a set of generalization trees $\mathcal{T}_k$ for $k=1, \dots, |Q_T|$, we define the \emph{generalization anonymity complex}~$\Gamma$ as the simplicial complex whose \k-simplices are determined by $(k+1)$ $|Q_T|$-dimensional tuples~$\{A_i\}_{i=1}^k$ such that there exists a $\bar{r}$ such that $\mathcal{T}_i(A_i,\bar{r}) = \mathcal{T}_j(A_j,\bar{r})$.
\end{definition}		

\medskip 

The definition of a generalized anonymity \k-simplex can be obtained as well. For example, $\{Portugal, Spain, Hungary\}$ forms a generalized anonymity 3-simplex for the generalization level 2, considering the tree in Figure~\ref{fig:ex_generalizations}.
	
\medskip

\begin{definition}[Generalization lattice]
	Given a $|Q_T|$-dimensional tuple representing attributes and the associated trees $\mathcal{T}_k$ for each attribute, we can construct a directed graph (lattice) where the vertices are the tuples $$<\mathcal{T}_1(A_1,s),\mathcal{T}_2(A_2,s_2),\dots,\mathcal{T}_{|Q_T|}(\gamma_{|Q_T|},s_{|Q_T|})>$$ and there is an edge between two vertices if there exists a generalization $s_i+1$ for an attribute $A_i$.
\end{definition}

Figure~\ref{fig:ex_lattice} shows an example of a generalization lattice for the attributes $\{Country, Gender\}$.

\medskip

Given the generalization trees and the previous definitions, we can see the similarity with the previous section, where we were interested in searching for regimes where the  anonymity complexes have trivial high-order homology groups. The critical difference is, however, that these generalizations are not dependent on a single parameter ($\epsilon$ in the previous section) anymore and thus we do not have a filtration in general.

In fact, the anonymity complexes, for various degrees of generalization, satisfy the following commutative diagram\footnote{We denote with $C^{00}$ the generalized anonymity complex corresponding to the generalization $<G_0,C_0>$, as in Figure~\ref{fig:ex_lattice}}:
\begin{figure}[h]
	\centering
	\includegraphics[height=2cm]{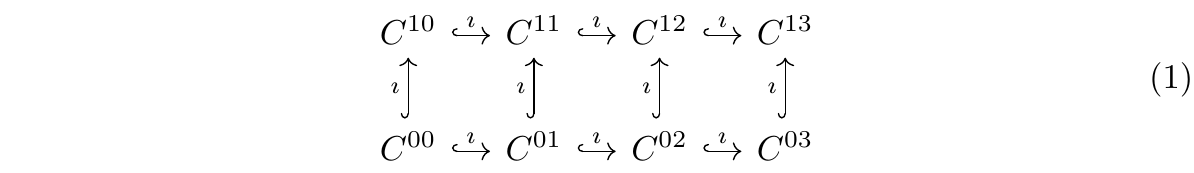}
\end{figure}
which is clearly not a filtration, but rather a \emph{multi-dimensional filtration}. Unfortunately, computation of the barcodes assoiated to a multi-dimensional persistent homology is still an open question, although progress has been made in recent years, see~\cite{KX-GWW:14} and references therein. 

There are two ways we can tackle this problem, using either some further assumptions on the generalization or an approximation. The first is via \emph{zig-zag persistent homology}~\cite{GC-VdS-MD:09}. In this context, we need to assume an ordering of the generalizations, something that is not uncommon, see~\cite{kim2013}. Given the structure of the problem, a reasonable sequence of spaces to be considered would be $C^{00},C^{01},C^{02},C^{03},C^{10},C^{11},C^{12},C^{13}$ leading to the following commutative diagram:
$$
	\underbrace{C^{00} \stackrel{\imath}{\hookrightarrow} C^{01} \stackrel{\imath}{\hookrightarrow} C^{02} \stackrel{\imath}{\hookrightarrow} C^{03}}_{C^{00-3}} \leftrightarrow \underbrace{C^{10} \stackrel{\imath}{\hookrightarrow} C^{11} \stackrel{\imath}{\hookrightarrow} C^{12} \stackrel{\imath}{\hookrightarrow} C^{13}}_{C^{10-3}}\,.
$$
Note that the middle map is not an inclusion. It is possible to still compute persistent homology of the full chain by ``joining'' the two subsequences of spaces $C^{00-3}$ and $C^{10-3}$ through their union as follows:
\begin{figure*}[h]
	\centering
	\includegraphics[height=1.4cm]{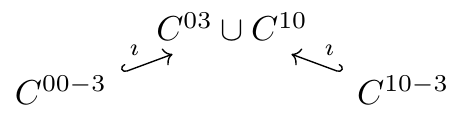}
\end{figure*}
In this case, by using Mayer-Vietoris pyramid principle~\cite{GC-VdS-MD:09}, we can compute the persistent homology of the chain and by passing through the union we obtain the barcodes for the initial chain~\cite{GC-VdS-MD:09}. Given the barcodes we are now in the same situation as in the previous section where we can search for different \k-anonymity regimes.

Note however that the previous barcode would be different than the one built for the following sequence, say $C^{00}, C^{10}, C^{01}, C^{11}, C^{02}, C^{12}, C^{03}, C^{13}$. In this case we would need to apply Mayer-Veitoris pyramid principle multiple times:
\begin{figure*}[h]
	\centering
	\includegraphics[height=4cm]{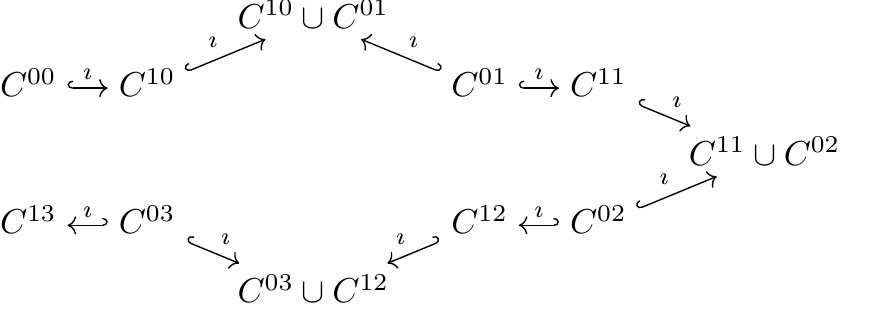}
\end{figure*}
	
An alternative approach to compute the barcodes for~\cd_lattice, is via an approximation. Specifically, one can consider the persistence equivalence theorem~\cite[Section  2.4.2]{Nanda:12} where, given the commutative diagram in~\cd_lattice, one can compute the persistent homology for the lower and upper chains. Because of the inclusion maps we have that the persistent homology modules of the two chains respect the inclusions. We can then proceed as follows: we first compute the persistent homology of the lower chain in the commutative diagram~\cd_lattice and if it achieves the desired \k-anonymity, then, because of the inclusion maps, such \k-anonymity can be achieved also by the upper chain in~\cd_lattice, thus we do not need to compute the barcodes for the upper chain and stop after obtaining the barcodes for the lower one.

However, if \k-anonymity is not achieved considering the lower chain, then persistent homology and corresponding barcodes need to be computed for the upper chain. Note that as we are approximating a multi-dimensional persistent homology by a sequence of (independent) persistent homology computation, if \k-anonymity cannot be achieved following such process we cannot conclude that there is not a way to \k-anonymize the data. This is a consequence of the natural complexity of the anonymity problem.

\section{Conclusion}
This paper introduced a new perspective to \k-anonymity in data privacy based on algebraic topology. In particular, we addressed the case when the data lies in a metric space. We demonstrated how tools such as persistence homology can be applied to efficiently obtain the entire spectrum of \k-anonymity of the database for various values of the radius of proximity. For this representation, we provided an analytic characterization of conditions under which a given representation of the dataset is \k-anonymous. Finally, we discussed how this method can be extended to address the general case of a mix of categorical and metric data. 

In future, it would be interesting to investigate other notions of privacy using these tools. In particular, applicability of such techniques to dynamic databases would be an interesting case which naturally arise in control applications. 

\section{Acknowledgments}

The authors would like to thank Vidit Nanda for discussions on zig-zag persistent homology and the Perseus software used in this paper to compute persistent homology.

\bibliographystyle{plain}
\bibliography{references_privacy}

\end{document}